\documentclass{article}
\usepackage{spconf,amsmath,graphicx,amsfonts,amsthm, algorithmic, algorithm}
\usepackage[caption=false]{subfig} 
\usepackage{multirow}
\usepackage{booktabs}

\providecommand{\norm}[1]{\lVert#1\rVert}
\newtheorem{theorem}{Theorem}[section]
\title{Analysis-based sparse reconstruction with synthesis-based solvers}
\name{
Nicolae Cleju $^{\star}$ \qquad
Maria G. Jafari$^{\dagger}$ \qquad 
Mark D. Plumbley$^{\dagger}$ 
\thanks{This work was supported by: EURODOC ``Doctoral Scholarships for research
performance at European level" project, financed by the European Social Found and Romanian
Government (NC), EU Framework 7 FET-Open project FP7-ICT-225913-SMALL: Sparse Models, Algorithms and Learning for Large-Scale data (MGJ) and a Leadership Fellowship (EP/G007144/1) from the UK Engineering and Physical Sciences Research Council (EPSRC) (MDP) }}

\address{$^{\star}$ Technical University ``Gheorghe Asachi" of Iasi, Romania \\ Faculty of Electronics, Telecommunications and Information Technology \\
    $^{\dagger}$ Queen Mary University of London, Centre for Digital Music, London, UK }

\begin{document}
%
\maketitle
\begin{abstract}
Analysis based reconstruction has recently been introduced as an alternative to the well-known synthesis sparsity model used in a variety of signal processing areas. In this paper we convert the analysis exact-sparse reconstruction problem to an equivalent synthesis recovery problem with a set of additional constraints. We are therefore able to use existing synthesis-based algorithms for analysis-based exact-sparse recovery. We call this  the Analysis-By-Synthesis (ABS) approach. We evaluate our proposed approach by comparing it against the recent Greedy Analysis Pursuit (GAP) analysis-based recovery algorithm. The results show that our approach is a viable option for analysis-based reconstruction, while at the same time allowing many algorithms that have been developed for synthesis reconstruction to be directly applied for analysis reconstruction as well.
\end{abstract}
\section{Introduction}
\label{sec:intro}

In recent years, sparse representation of signals has been an active research domain in signal processing. Until recently the usual sparsity model considered was a generative model, known as \emph{synthesis sparsity}: a signal $x \in \mathbb{R}^d$ is sparse if it can be expressed as a weighted sum of a few signals (called \emph{atoms}) from a known dictionary $D \in \mathbb{R}^{d \times N}$
\begin{equation}
\label{eq_smodel}
x = D \gamma_S, \textrm{ with } \norm{\gamma_S}_0 = k
\end{equation}
where $\norm{\cdot}_0$ represents the $\ell_0$ pseudo-norm, defined as the number of non-zero coefficients of a vector. The decomposition vector $\gamma_S$ is thus required to have $k$ non-zero elements.

Lately, a different sparsity model known as \emph{analysis sparsity} has been proposed \cite{Elad2007}, asserting that the signal $x$ produces a sparse output
\begin{equation}
\label{eq_amodel}
\gamma_A = \Omega x, \textrm{ with } \norm{\gamma_A}_0 = N - l
\end{equation}
when analyzed with an operator $\Omega \in R^{N \times d}$, where $l$ is the number of zero coefficients of $\gamma_A$ (see Section \ref{sec:analysis} for more details).

Both of these models can be successfully used as regularizing terms for ill-posed inverse problems. In this paper we focus on reconstructing a signal $x$ that is observed only through a set of $m < d$ linear measurements, arranged as the rows of an acquisition matrix $M \in \mathbb{R}^{m \times d}$, possibly affected by noise $e$
\begin{equation}
\label{eq_CSacq}
y = M x + e.
\end{equation}
This is known as the compressed sensing problem, which has been extensively studied \cite{CSStableSigRec, DecodingByLPCandes2005} and used in practice in various applications (e.g. \cite{Fira2011EUSPICO, Cleju2011ISSCS}). It is now well known \cite{CSStableSigRec} that a sufficiently sparse signal $x$ can be efficiently recovered from the measurements $y$ by solving the synthesis-based optimization problem
\begin{equation}
\label{eq_CSrec_s}
\hat{x} = D \arg\min_{\gamma_S} \| \gamma_S \|_0 \textrm{ with } \| y - M D \gamma_S\|_2^2 < \epsilon
\end{equation}
where $\epsilon$ is the estimated noise energy of the measurements. Interestingly, it has also been shown \cite{Nam2011techreport} that a sufficiently sparse $\gamma_A$ in \eqref{eq_amodel} also allows accurate recovery of the signal $x$ by solving
\begin{equation}
\label{eq_CSrec_a}
\hat{x} = \arg\min_{x} \| \Omega x \|_0 \textrm{ with } \| y - M x\|_2^2 < \epsilon.
\end{equation}
Problem \eqref{eq_CSrec_s} is NP-complete \cite{proofNPhard}, and we hypothesize that \eqref{eq_CSrec_a} may be similarly difficult. However, under stricter conditions, the $\ell_0$ norm in both equations can be replaced with the $\ell_1$ norm, leading to convex optimization problems that are much easier to solve \cite{Candes2011}.

In this paper we pursue an approach to solving the analysis-based reconstruction problem \eqref{eq_CSrec_a} by rewriting it as an equivalent synthesis reconstruction problem. The paper is structured as follows. Section \ref{sec:analysis} contains a closer look at the analysis model and its details. In Section \ref{sec:reform} we propose the \emph{Analysis-By-Synthesis} (ABS) scheme for exact-sparse analysis recovery. In Section \ref{sec:results} we compare our approach with the results obtained with the Greedy Analysis Pursuit algorithm \cite{Nam2011ICASSP} designed to solve \eqref{eq_CSrec_a} directly. Section \ref{sec:disc} contains further considerations regarding our scheme. Finally, concluding remarks and future work are presented in Section \ref{sec:concl}.

\section{The analysis model}
\label{sec:analysis}

As shown in \eqref{eq_smodel} and \eqref{eq_amodel}, the synthesis sparsity model requires that a signal is composed out of only $k$ atoms (columns) of the dictionary $D$, whereas the analysis model requires the signal to be orthogonal to a large number $l$ of the rows of the operator $\Omega$. We refer to $k$ as the \emph{sparsity} of the signal $x$ in the dictionary $D$, and, following \cite{Nam2011techreport}, to $l$ as the \emph{cosparsity} of $x$ with respect to the operator $\Omega$. 


For $N \le d$ the analysis and synthesis reconstruction problems are shown in \cite{Elad2007} to be equivalent, with $D$ and $\Omega$ being pseudo-inverses to each other, $D = \Omega ^ \dagger$. However, in general for $N > d$ (i.e. $\Omega$ is a ``tall" matrix) the equivalence no longer holds, with \eqref{eq_CSrec_s} and \eqref{eq_CSrec_a} leading to different solutions. 

The similarity of the two models emerges from the fact that both are instances of the \emph{Union-of-Subspaces} (UoS) model \cite{Nam2011techreport}. The set of all $k$-sparse signals in a dictionary $D$ comprises the union of all the $\binom{N}{k}$ $k$-dimensional subspaces spanned by any subset of $k$ atoms from the $N$ atoms of $D$. The set of all $l$-cosparse signals of an operator $\Omega$ is the union of all the $\binom{N}{l}$ $(d-l)$-dimensional subspaces that are the orthogonal complements of the subspaces spanned by any $l$ rows. We may say, therefore, that the synthesis model is essentially described by the subspaces where the signal may lie, i.e. the non-zero coefficients of the decomposition, whereas the analysis model describes the subspaces where the signal cannot lie, i.e. the rows that are orthogonal to the signal \cite{Nam2011techreport}. 

\section{Analysis-By-Synthesis (ABS) approach for exact recovery}
\label{sec:reform}

\subsection{Augmented equivalence theorem}

Let us consider the case of reconstruction with exact constraints, i.e. $\epsilon = 0$ in  \eqref{eq_CSrec_a}. The following theorem establishes the equivalence between the analysis recovery problem and a synthesis recovery problem with a set of extra constraints.

\begin{theorem}
\label{th_equiv}
The solution of the analysis recovery problem with exact constraints and full-rank operator $\Omega$
\begin{equation}
\label{eq_CSrec_a_exact}
\hat{x} = \arg\min_{x} \| \Omega x \|_0 \textrm{ with } y = M x
\end{equation}
is identical to the solution of the augmented synthesis recovery problem
\begin{equation}
\label{eq_CSrec_s_exact}
\hat{x} = D \arg\min_{\gamma} \| \gamma \|_0 \textrm{ with } \tilde{y} = \tilde{A} \gamma 
\end{equation}
where $D = \Omega^\dagger$, $\tilde{y} = \left[ \begin{array}{c} y \\ 0 \end{array} \right]$, $\tilde{A} = \left[ \begin{array}{c} M D \\ P_D \end{array} \right]$ with $P_D$ being any projector on the nullspace of $D$.
\end{theorem}
\begin{proof}
We show the equivalence of \eqref{eq_CSrec_a_exact} with \eqref{eq_CSrec_s_exact}, starting from the approach in \cite{Elad2007}. Making the notation $\Omega x = \gamma$, it follows from $\Omega^\dagger \Omega = I_d$ that $x = \Omega^\dagger \gamma$. We proceed to substitute the unknown variable $x$ in \eqref{eq_CSrec_a_exact} introducing $\gamma$ instead, but in doing that we must keep in mind that $\gamma$ is allowed to live only in the column span of $\Omega$, which we can express as the extra constraint $\gamma = \Omega \Omega^\dagger \gamma$. Therefore we arrive to
\begin{equation}
\label{eq_CSrec_a_reform1}
\hat{x} = \Omega^\dagger \arg\min_{\gamma: \gamma = \Omega \Omega^\dagger \gamma} \| \gamma \|_0 \textrm{ with } y = M \Omega^\dagger \gamma.
\end{equation}
We rewrite the constraint $\gamma = \Omega \Omega^\dagger \gamma$ as $0 = (I_N - \Omega \Omega^\dagger ) \gamma$. We can join this with the constraint $y = M \Omega^\dagger \gamma$ and construct a single augmented constraint system

\begin{equation}
\label{eq_extsys}
\underbrace{ \left[ \begin{array}{c} y \\ 0 \end{array} \right]}_{\tilde{y}} = \underbrace{ \left[ \begin{array}{c} M \Omega^\dagger \\ I_N - \Omega \Omega^\dagger \end{array} \right] }_{\tilde{A}} \gamma.
\end{equation}
Let us define $D = \Omega^\dagger$. The lower constraint $0 = (I_N - \Omega \Omega^\dagger) \gamma$ is equivalent to $\gamma$ living in the column space of $\Omega$, i.e. being orthogonal to the nullspace of $D = \Omega^\dagger$ (denoted as $n_D$); therefore this constraint can be expressed as $0 = P_D \gamma$ with $P_D$ being any projector on $n_D$.
Replacing $\Omega^\dagger$ with $D$ and rewriting \eqref{eq_CSrec_a_reform1} with the augmented constraint \eqref{eq_extsys} yields
\begin{equation}
\label{eq_CSrec_s_exact_thfinal}
\hat{x} = D \arg\min_{\gamma} \| \gamma \|_0 \textrm{ with } \left[ \begin{array}{c} 	y \\ 0 \end{array} \right] = \left[ \begin{array}{c}  M D \\	P_D \end{array} \right] \gamma 
\end{equation}
which is what we wanted to prove.
\end{proof}

Theorem \ref{th_equiv} reveals that analysis recovery is a particular instance of synthesis recovery; indeed, without the lower constraint \eqref{eq_CSrec_s_exact_thfinal} would be identical to synthesis-based recovery. What is specific of the analysis recovery is, therefore, the restriction of the solution search space to the column space of $\Omega$ (or, equivalently, to the row space of $D = \Omega^\dagger$). A similar condition is used in \cite{Yaghoobi2011} in the context of local optimality of analysis operator learning. In practice, this constraint can be expressed by finding a set of $(N - d)$ linearly independent vectors from $n_D$ (e.g. by finding a SVD decomposition of D) and then imposing that $\gamma$ is orthogonal to all of the vectors in this set. Moreover, if $D$ is a tight frame allowing fast multiplications via fast transform algorithms, the row vectors of $P_D$ can be selected as the ``missing'' orthogonal rows, thus allowing possible fast solver implementations. 

As a consequence of Theorem \ref{th_equiv}, one can use synthesis-based solvers to find the solution for analysis-based recovery. While the more general character of synthesis over analysis recovery, as well as the subspace restrictions implied by the latter, is already known \cite{Elad2007, Nam2011techreport}, to our knowledge this is the first time that the equivalence of analysis exact reconstruction with an augmented synthesis problem has been stated explicitly and also used as a method for analysis recovery.

One observes that whenever $N \le d$, $\Omega \Omega^\dagger = I_N$ and thus the lower subspace constraint in \eqref{eq_extsys} vanishes, straightforwardly confirming the equivalence of analysis-based and synthesis-based recovery already shown in \cite{Elad2007} for this case.

\subsection{Proposed approach}

Our proposed approach for analysis recovery with exact constraints ($\epsilon \approx 0$) is summarized in Algorithm \ref{algo_ABS}, which we denote as \emph{Analysis-By-Synthesis} (ABS). It consists of building the augmented constraint matrix $\tilde{A}$ and measurement vector $\tilde{y}$ and then solving with a synthesis-based algorithm.

\floatname{algorithm}{Algorithm}
\begin{algorithm}[t]
\caption{Proposed Analysis-By-Synthesis (ABS) approach for exact reconstruction}
\label{algo_ABS}
\begin{algorithmic}[1]
  \REQUIRE Analysis operator $\Omega$, measurements vector $y$, measurement matrix $M$\\
  \ENSURE  Recovered signal \[\displaystyle \hat{x} = \arg\min_{x} \| \Omega x \|_0 \textrm{ with } y = M x\]
  \STATE Define $D = \Omega^\dagger$ and compute a basis for the null space of $D$ using the $SVD$ decomposition, arranging the vectors as the rows of a $(N-d) \times N$ matrix denoted as $P_D$
  \STATE  Create augmented constraint matrix $\tilde{A}$ and measurement vector $\tilde{y}$ \\
  \[ \tilde{A} = \left[ \begin{array}{c} M D \\	P_D \end{array} \right] \;\;\;\;\;\;\;\;\;\;\;\;\;
     \tilde{y} = \left[ \begin{array}{c} y   \\ 0 \end{array} \right]
  \]
  \STATE Solve \[\displaystyle \hat{x} = D \arg\min_{\gamma} \| \gamma \|_0 \textrm{ with } \tilde{y} = \tilde{A} \gamma\] \\ using a synthesis-based solver.
\end{algorithmic}
\end{algorithm}

\section{Experimental results}
\label{sec:results}

\subsection{Setup}

A significant advantage of our approach is the ability to use existing $\ell_0$ or $\ell_1$ solvers designed for the synthesis reconstruction problem, in the third step of Algorithm \ref{algo_ABS}. We run four different synthesis-based solvers in the proposed ABS approach: Orthogonal Matching Pursuit (OMP) \cite{Pati1993} with stopping criterion being a fixed number $k$ of selected atoms (denoted as OMP-$k$), OMP with stopping criterion being residual error below $10^{-9}$ (OMP-$\epsilon$), Two Stage Thresholding (TST) \cite{Maleki2010} (a generalization of CoSaMP and subspace pursuit) and Basis Pursuit (BP) \cite{DecodingByLPCandes2005} for $\ell_1$ minimization from \cite{l1magic}. For reference we compare with the results obtained with the Greedy Analysis Pursuit (GAP) \cite{Nam2011ICASSP} algorithm, which is specifically designed for solving the analysis recovery problem directly.

We investigate the phase transition border \cite{Nam2011ICASSP} of the above mentioned algorithms for perfect recovery, for the case of exact reconstruction. The dimension of the signals is fixed to 200. The analysis operator is created as the transposition of a random tight frame, having $N = 240$ rows. We define the parameters $\delta = \frac{m}{d}$ and $\rho = \frac{d-l}{m}$ that define the compression ratio and the relative cosparsity. For every pair $(\delta, \rho)$ we generate 100 signals $x_i$ such that $\| \Omega x_i \|_0 = N-l$ and we project them using a random measurement matrix $M$ of size $m \times d$, with zero-mean unit-norm normal i.i.d. random elements. We then attempt reconstruction with the above mentioned algorithms. For OMP-$k$ we stop after $k = N-l$ atoms have been selected. We consider a signal as perfectly recovered if the reconstruction error is below $10^{-6}$. 

\subsection{Results}

Fig.\ref{fig:results} displays the percentage of perfectly recovered signals, with white indicating 100\% recoverability and black 0\%. The notation ABS indicates that the synthesis solvers are used within our proposed approach.

\begin{figure*}
  \centering
  \subfloat[ABS: OMP-k]
  {
		\includegraphics[scale=0.3]{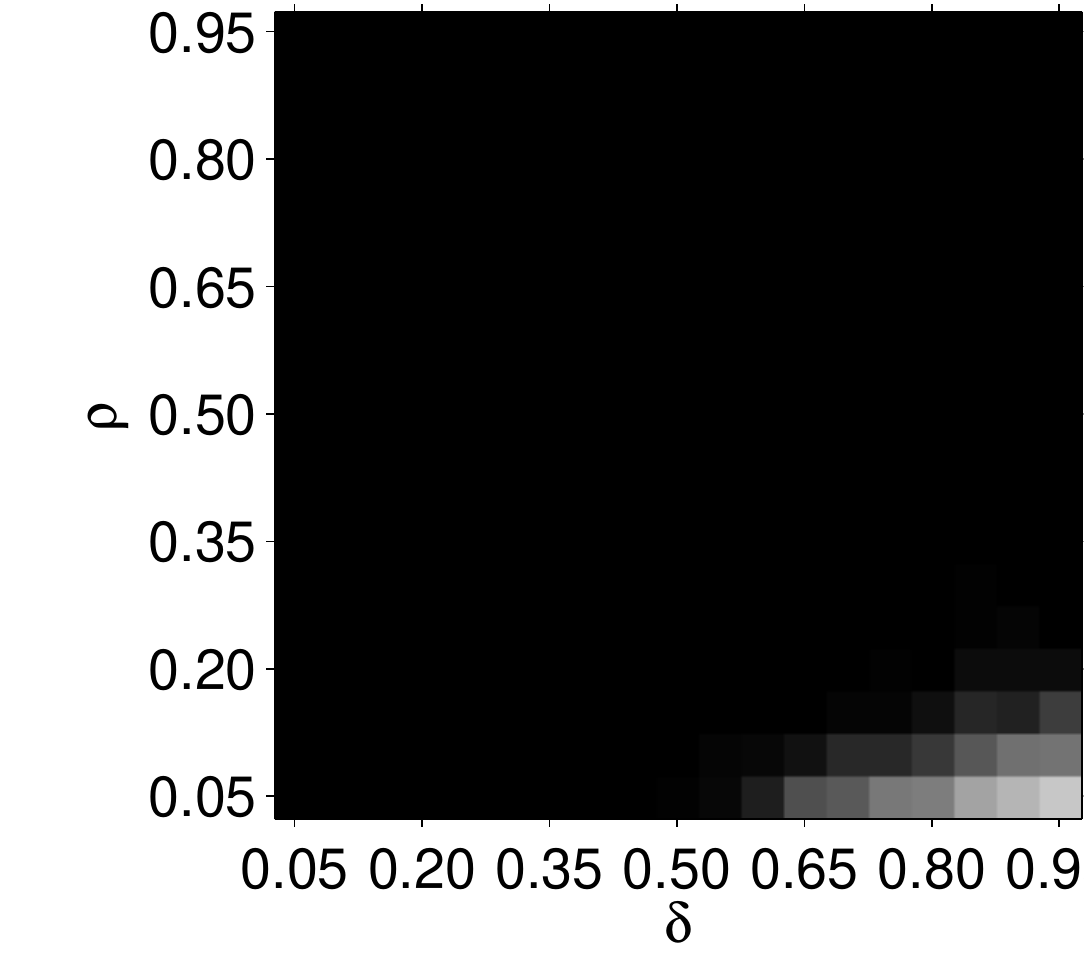}
		\label{fig:OMPA}
	}
  \subfloat[ABS: OMP-$\epsilon$]
  {
		\includegraphics[scale=0.3]{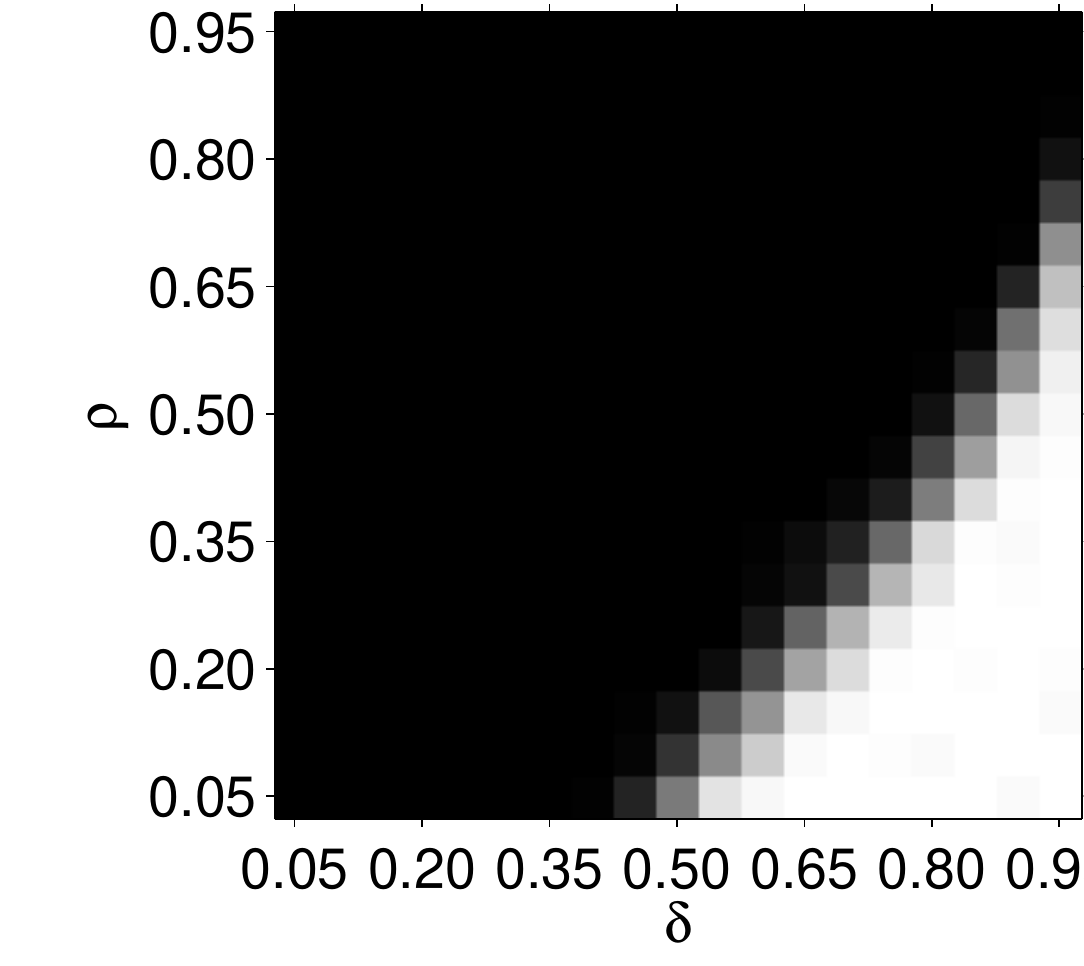}
		\label{fig:OMPepsA}
	}
  \subfloat[ABS: TST]
  {
		\includegraphics[scale=0.3]{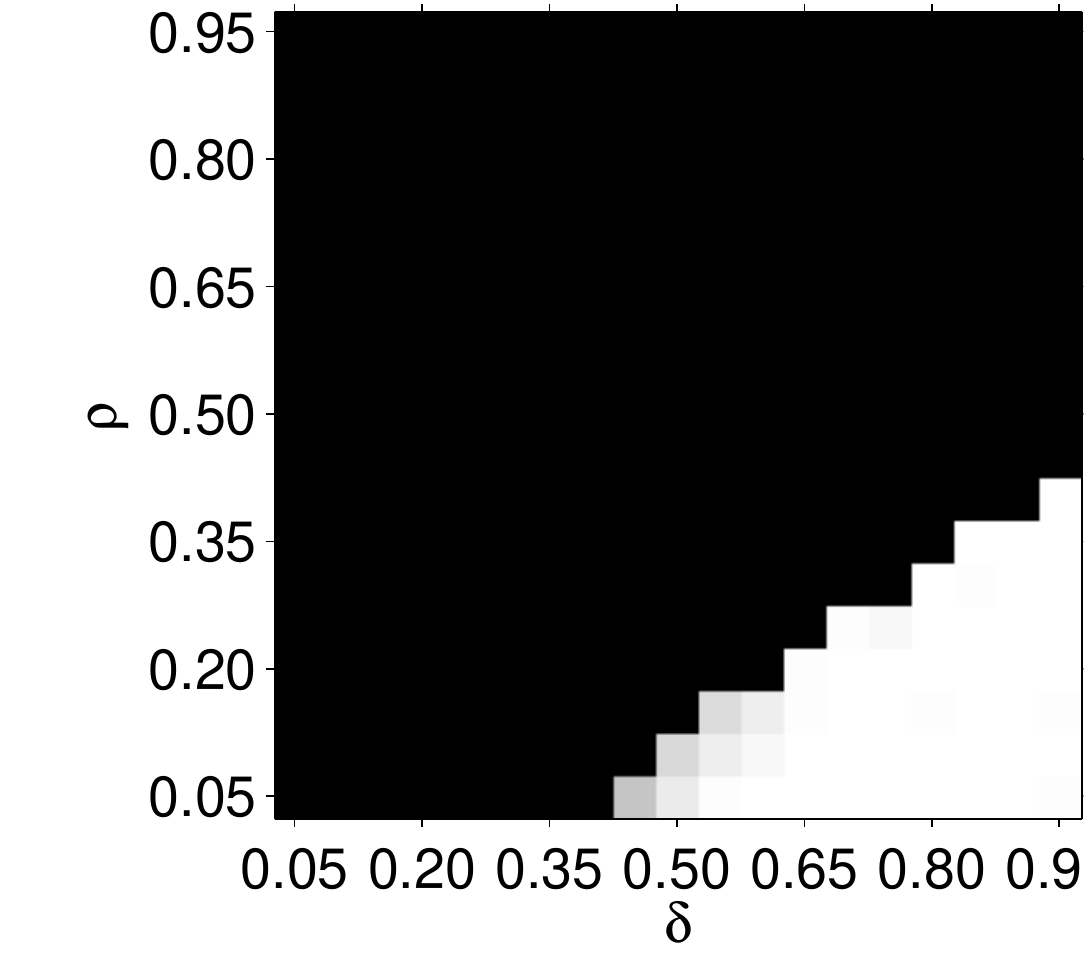}
		\label{fig:TSTA}
	}
  \subfloat[ABS: BP]
  {
		\includegraphics[scale=0.3]{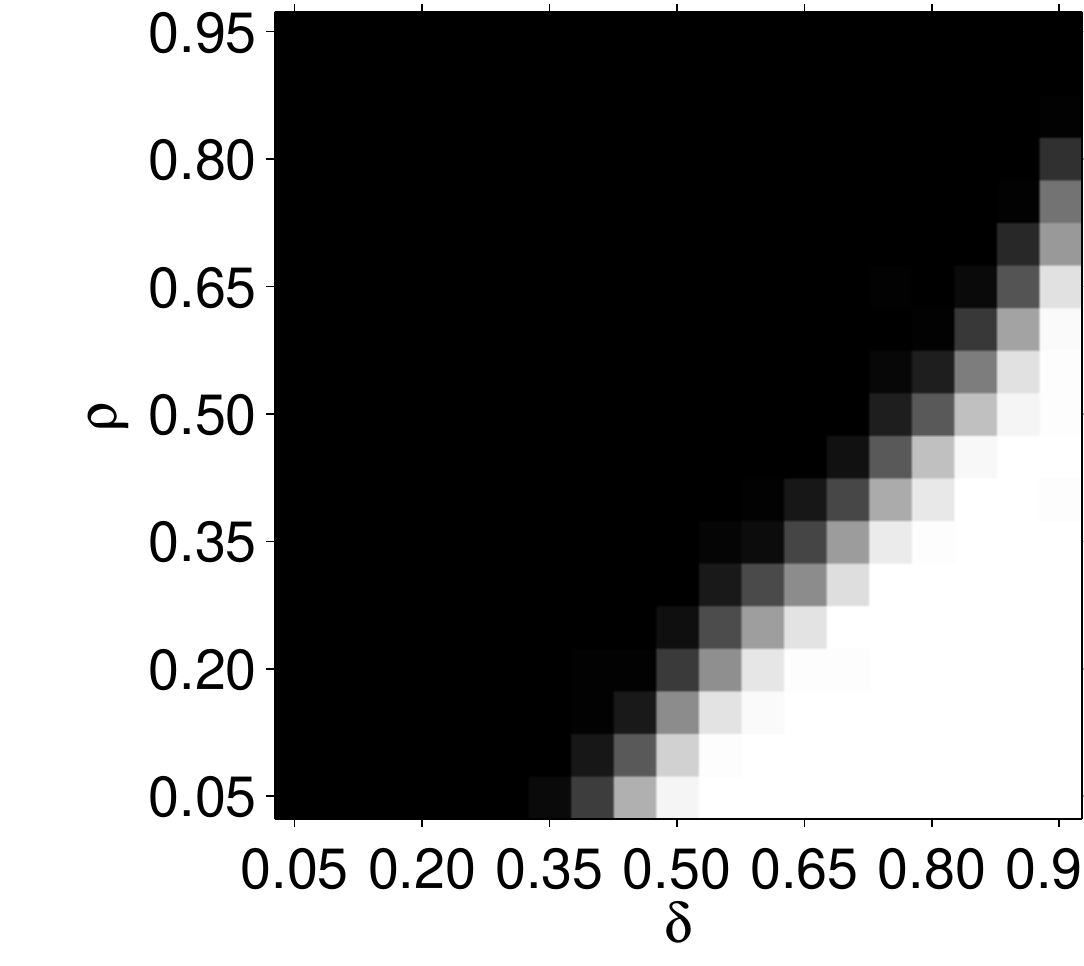}
		\label{fig:BPA}
	}
  \subfloat[GAP]
  {
		\includegraphics[scale=0.3]{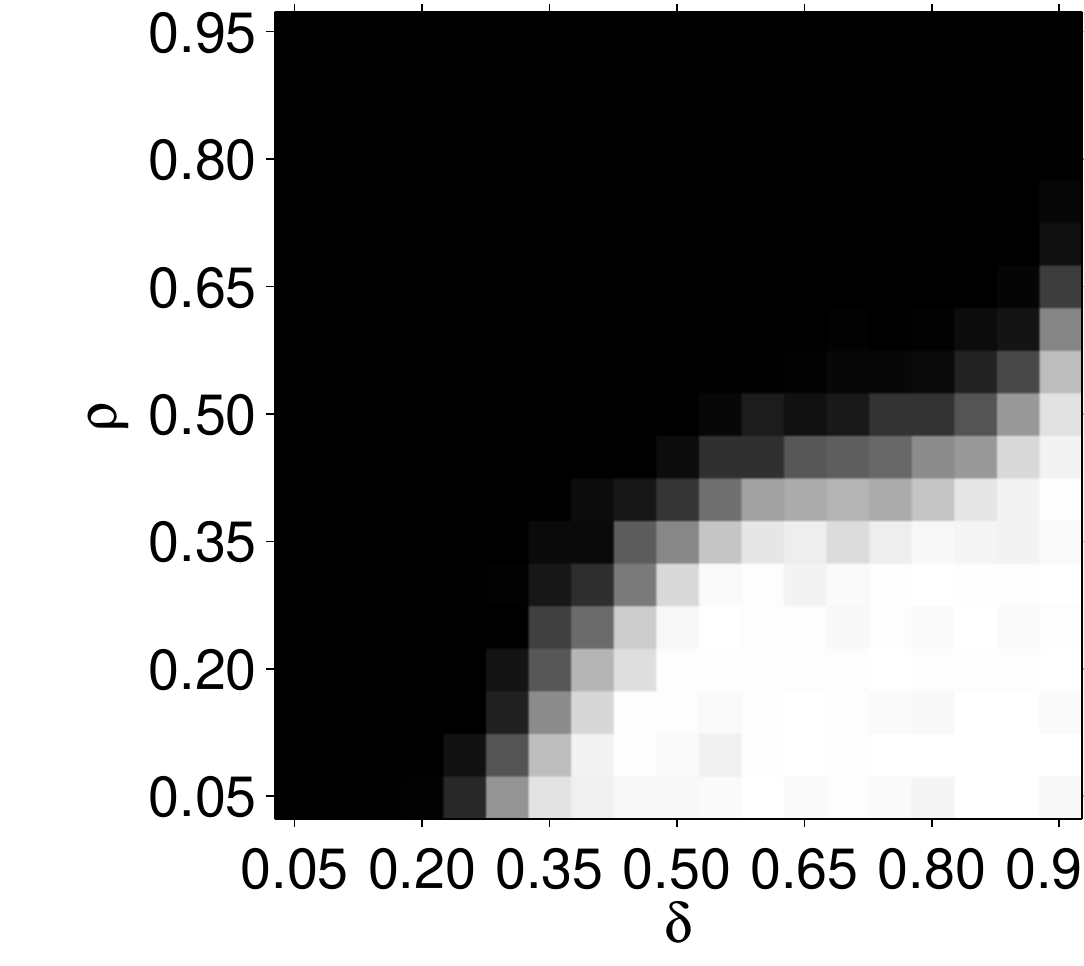}
		\label{fig:GAP}
	}
  \hfill
\caption{Percentage of perfectly reconstructed signals for analysis-based recovery with different algorithms: our proposed ABS approach with four different solvers ( (a), (b), (c) and (d) ) and the GAP algorithm \cite{Nam2011ICASSP} (e). White indicates 100\% recoverability and black 0\%.}
\label{fig:results}
\end{figure*}

The results show that our approach is a viable solution to analysis-based recovery. However, we find that not all synthesis solvers are adequate for use with our approach: OMP-$k$ performs poorly, suggesting that this should not be considered as an option for recovery. OMP-$\epsilon$, TST and BP provide good results, with OMP-$\epsilon$ and BP outperforming GAP in some areas (lower cosparsity but sufficient measurements, i.e. larger $\delta$ and larger $\rho$), but being outperformed in others (fewer measurements and higher cosparsity, i.e. smaller $\delta$ and $\rho$).

For completeness, we also present the total running times of the algorithms in Table \ref{table_times}. The overall experiment consisted in recovering a total of 36100 signals ( $19 \times 19$ pairs $(\delta, \rho)$ $\times 100$ signals for each) on a $2.83$GHz Intel Core 2 Quad Q9550 machine running MATLAB 7.7.0. We find that for our experiments OMP recovery is the fastest whereas BP is the slowest, with TST and GAP yielding intermediate times. 

\section{Discussion and further considerations}
\label{sec:disc}

As we have seen, the proposed approach is based on rewriting analysis recovery as a particular case of the more general synthesis recovery problem, subsequently applying a general synthesis-based solver. Therefore the solver may not fully exploit the particularities of the analysis problem, reflected in the particular structure of the augmented constraint matrix $\tilde{A}$ (i.e. the bottom rows are orthogonal to the upper rows). This makes it possible for the results not to be as good as with a solver designed exclusively for analysis-based reconstruction. For the purpose of this paper, however, we settle with the possible slight suboptimality of the synthesis solvers, counterbalanced by the increased flexibility conferred by the large number of available solvers.

For reconstruction with approximate constraints, i.e. $\epsilon \ge 0$ in \eqref{eq_CSrec_a}, an extra precaution is required when handling the augmented constraint matrix in \eqref{eq_extsys}. We still require that the solution $\gamma$ satisfies the lower subspace constraint as precisely as possible, but we allow a certain degree of approximation error for the upper part. This prevents a direct application of Theorem \ref{th_equiv}. We are currently working on establishing a similar equivalence relation for the case of approximate recovery. 

\begin{table}
\caption{Total running times ($\times 10^3$ seconds)}
  \begin{center}
    \begin{tabular}{ccccc}
      \toprule
        \multicolumn{4}{c}{ABS:} \\
	    OMP-$k$ & OMP-$\epsilon$ & TST & BP & GAP \\
      \cmidrule(lr){1-4} \cmidrule(lr){5-5}
	    7.691 & 8.004 & 22.829 & 51.152 & 13.065 \\
      \bottomrule
    \end{tabular}
  \end{center}
\label{table_times} 
\end{table}


\section{Conclusions and future work}
\label{sec:concl}

In this paper we have presented a new approach to analysis-based exact signal recovery, by reformulating it as a particular synthesis-based problem. We prove that, for reconstruction with exact constraints, analysis recovery is equivalent to synthesis recovery with an augmented constraint matrix. This means that we can use synthesis-based algorithms for analysis recovery. Experimental results show that our approach is a viable alternative for analysis-based reconstruction.

For future work, it will be interesting to investigate which algorithms are suitable for this approach and the reason why some, such as OMP-$k$, are performing poorly, while others provide good results. We also aim to extend the equivalence for approximate recovery.

\bibliographystyle{IEEEtran}
\bibliography{IEEEabrv,JabRef_short}

\end{document}